\documentclass[11pt]{article}
\usepackage{a4wide}
\usepackage[usenames,dvipsnames]{color}
\usepackage{amssymb,amsmath,graphics,color,enumerate,eucal}
\usepackage{tikz}
\usepackage{xspace}
\usetikzlibrary{arrows,shapes,snakes,automata,backgrounds,petri,patterns} 
\usepackage{todonotes}
\usepackage[utf8x,utf8]{inputenc} 
\usepackage[T1]{fontenc}
\usepackage[vlined,linesnumbered,boxruled]{algorithm2e}

\usepackage{booktabs}
\usepackage{array}

\usepackage{standalone}
\usepackage{tikz}
\usetikzlibrary{arrows}
\usetikzlibrary{patterns}

\usepackage{epsfig}
\usepackage{subfig}
\usepackage{amsmath,amsfonts,amssymb,epsfig,color,amsthm}
\usepackage{comment}
\usepackage{thmtools}
\usepackage{thm-restate}
\usepackage{mdframed}

\newtheorem{definition}{Definition}[]
\newtheorem{theorem}{Theorem}[]
\newtheorem{lemma}[theorem]{Lemma}
\newtheorem{corollary}[theorem]{Corollary}

\newtheorem{proposition}[theorem]{Proposition}

\newcommand{\floor}[1]{\ensuremath{\left\lfloor{#1}\right\rfloor}}%
\newcommand{\ceil}[1]{\ensuremath{\left\lceil{#1}\right\rceil}}%
\newcommand{\ignore}[1]{}%

\newcommand{\ProblemFormat}[1]{{\sc #1}}
\newcommand{\ProblemName}[1]{\ProblemFormat{#1}\xspace}

\newcommand{\mycase}[1]{\noindent{\bf CASE #1\ }}

\newcommand{\gain}{{\rm gain}}
\newcommand{\tw}{{\rm tw}}
\newcommand{\pw}{{\rm pw}}

\newcommand{\probAPSP}{\ProblemName{All Pairs Shortest Paths}}
\newcommand{\probKOPTOpt}{\ProblemName{$k$-opt Optimization}}
\newcommand{\probKOPTDec}{\ProblemName{$k$-opt Detection}}
\newcommand{\probFiveOPTOpt}{\ProblemName{$5$-opt Optimization}}
\newcommand{\probOPTOpt}[1]{\ProblemName{$#1$-opt Optimization}}
\newcommand{\probOPTDec}[1]{\ProblemName{$#1$-opt Detection}}
\newcommand{\probNegativeTriangle}{\ProblemName{Negative Edge-Weighted Triangle}}
\newcommand{\heading}[1]{\medskip\noindent{\bf #1.\ }}%

\newcommand{\Mm}{{\ensuremath{\mathcal{M}}}}

\newcommand{\Tt}{{\ensuremath{\mathcal{T}}}}

\DeclareMathOperator{\polylog}{polylog}

\begin{document}
\title{Improving TSP tours using dynamic programming over tree decompositions\thanks{The work of M. Cygan and \L. Kowalik is a part of the project TOTAL that has received funding from the European Research Council (ERC) under the European Union’s Horizon 2020 research and innovation programme (grant agreement No 677651).
 A. Soca{\l}a is supported by the National Science Centre of Poland, grant number 2013/09/B/ST6/03136.} 
}

\date{}

\author{Marek Cygan \quad \L ukasz Kowalik \quad Arkadiusz Soca\l a \\
{\small Institute of Informatics, University of Warsaw, Poland}\\
{\small \texttt{\{cygan,a.socala,kowalik\}@mimuw.edu.pl}}}

\maketitle

\begin{abstract}
Given a traveling salesman problem (TSP) tour $H$ in graph $G$ a $k$-move is an operation which removes $k$ edges from $H$, and adds $k$ edges of $G$ so that a new tour $H'$ is formed. 
The popular $k$-OPT heuristic for TSP finds a local optimum by starting from an arbitrary tour $H$ and then improving it by a sequence of $k$-moves.

Until 2016, the only known algorithm to find an improving $k$-move for a given tour was the naive solution in time $O(n^k)$. 
At ICALP'16 de~Berg, Buchin, Jansen and Woeginger showed an $O(n^{\floor{2/3k}+1})$-time algorithm.

We show an algorithm which runs in $O(n^{(1/4+\epsilon_k)k})$ time, where $\lim_{k\rightarrow\infty} \epsilon_k = 0$. 
It improves over the state of the art for every $k\ge 5$.
For the most practically relevant case $k=5$ we provide a slightly refined algorithm running in $O(n^{3.4})$ time.
We also show that for the $k=4$ case, improving over the $O(n^3)$-time algorithm of de~Berg et al. would be a major breakthrough: an $O(n^{3-\epsilon})$-time algorithm for any $\epsilon>0$ would imply an $O(n^{3-\delta})$-time algorithm for the \probAPSP problem, for some $\delta>0$.
\end{abstract}

\section{Introduction}

In the Traveling Salesman Problem (TSP) one is given a complete graph $G=(V,E)$ and a weight function $w:E\rightarrow\mathbb{N}$.
The goal is to find a Hamiltonian cycle in $G$ (also called a {\em tour}) of minimum weight.
This is one of the central problems in computer science and operation research.
It is well known to be NP-hard and has been researched from different perspectives, most notably using approximation~\cite{arora,christofides,SeboV14}, exponential-time algorithms~\cite{held-karp,karp} and heuristics~\cite{PadbergR91,LinK73,croes1958}.

In practice, TSP is often solved by means of local search heuristics where we begin from an arbitrary Hamiltonian cycle in $G$, and then the cycle is modified by means of some local changes in a series of steps.
After each step the weight of the cycle should improve; when the algorithm cannot find any improvement it stops.
One of the most successful examples of this approach is the $k$-opt heuristic, where in each step an improving $k$-move is performed. 
Given a Hamiltonian cycle $H$ in a graph $G=(V,E)$ a {\em $k$-move} is an operation that removes $k$ edges from $H$ and adds $k$ edges of $G$ so that the resulting set of edges $H'$ is a new Hamiltonian cycle. 
The $k$-move is {\em improving} if the weight of $H'$ is smaller than the weight of $H$.
The $k$-opt heuristic has been introduced in 1958 by Croes~\cite{croes1958} for $k=2$, and then applied for $k=3$ by Lin~\cite{lin1965computer} in 1965.
Then in 1972 Lin and Kernighan designed a complicated heuristic which uses $k$-moves for unbounded values of $k$, though restricting the space of $k$-moves to search to so-called sequential $k$-moves. A variant of this heuristic called LKH, implemented by Helsgaun~\cite{Helsgaun00}, solves optimally instances up to 85\,900 cities.
Among other modifications, the variant searches for non-sequential 4- and 5-moves.
From the theory perspective, the quality of the solutions returned by $k$-opt, as well as the length of the sequence of $k$-moves needed to find a local optimum, was studied, among others, by Johnson, Papadimitriou and Yannakakis~\cite{jpy88}, Krentel~\cite{krentel} and Chandra, Karloff and Tovey~\cite{chandra99}.
More recently, smoothed analysis of the running time and approximation ratio was investigated by Manthey and Veenstra~\cite{smoothed1} and K{\"{u}}nnemann and Manthey~\cite{smoothed2}.

In this paper we study the $k$-opt heuristic but we focus on its basic ingredient, namely on finding a single improving $k$-move.
The decision problem \probKOPTDec is to decide, given a tour $H$ in an edge weighted complete graph $G$, if there is an improving $k$-move.
In its optimization version, called \probKOPTOpt, the goal is to find a $k$-move that gives the largest weight improvement, if any.
Unfortunately, this is a computationally hard problem.
Namely, Marx~\cite{marx08} has shown that \probKOPTDec is $W[1]$-hard, which means that it is unlikely to be solvable in $f(k)n^{O(1)}$ time, for any function $f$.
Later Guo, Hartung, Niedermeier and Such{\'{y}}~\cite{GuoHNS13} proved that there is no algorithm running in time $n^{o(k/\log k)}$, unless Exponential Time Hypothesis (ETH) fails.
This explains why in practice people use exhaustive search running in $O(n^k)$ time for every fixed $k$, or faster algorithms which explore only a very restricted subset of all possible $k$-moves.

Recently, de~Berg, Buchin, Jansen and Woeginger~\cite{BergBJW16} have shown that it is possible to improve over the naive exhaustive search.
For every fixed $k\ge 3$ their algorithm runs in time $O(n^{\floor{2k/3}+1})$ and uses $O(n)$ space. 
In particular, it gives $O(n^3)$ time for $k=4$. 
Thus, the algorithm of de~Berg et al.\ is of high practical interest: the complexity of the $k=4$ case now matches the complexity of $k=3$ case, and hence it seems that one can use 4-opt in all the applications where 3-opt was fast enough.
De~Berg et al.\ show also that a progress for $k=3$ is unlikely, namely they show that \probKOPTDec has an $O(n^{3-\epsilon})$-time algorithm for some $\epsilon>0$ iff {\sc All Pairs Shortest Paths} problem can be solved in $O(n^{3-\delta})$-time algorithm for some $\delta>0$.

\heading{Our Results} 
In this paper we extend the line of research started in~\cite{BergBJW16}: we show an algorithm running in time $O(n^{(1/4+\epsilon_k)k})$ and using space $O(n^{(1/8+\epsilon_k)k})$ for every fixed $k$, where $\lim\epsilon_k=0$. 
We are able to compute the values of $\epsilon_k$ for $k\le 10$.
These values show that our algorithm improves the state of the art for every $k=5,\ldots,10$ (see Table~\ref{tab:comparision}).
A different adjustment of parameters of our algorithm results in time $O(n^{k/2+3/2})$ and additional space of $O(\sqrt{n})$, which improves the state of the art for every $k \ge 8$.

We also show a good reason why we could not improve over the $O(n^3)$-time algorithm of de~Berg et al. for \probOPTOpt{4}: an $O(n^{3-\epsilon})$-time algorithm for some $\epsilon>0$ would imply that {\sc All Pairs Shortest Paths} can be solved in time $O(n^{3-\delta})$ for some $\delta>0$. Note that although the family of $4$-moves contains all $3$-moves, it is still possible that there is no improving 3-move, but there is an improving 4-move. Thus the previous lower bound of de~Berg et al.\ does not imply our lower bound, though our reduction is essentially an extension of the one by de~Berg et al.~\cite{BergBJW16} with a few additional technical tricks. 

We also devote special attention to the $k=5$ case of \probKOPTOpt problem, hoping that it can still be of a practical interest.
Our generic algorithm works in $O(n^{3.67})$ time in this case.
However, we show that in this case the algorithm can be further refined, obtaining the $O(n^{3.4})$ running time.
We suppose that similar improvements of order $n^{\Omega(1)}$ should be also possible for larger values of $k$.
In Table~\ref{tab:comparision} we present the running times for $k=5,\ldots,10$.

\begin{table}
\begin{center}
\begin{tabular}{l|c|c|c|c|c|c}
$k$ & 5 & 6 & 7 & 8 & 9 & 10\\
\midrule
previous algorithm~\cite{BergBJW16}     & $O(n^{4})$ & $O(n^5)$ & $O(n^5)$ & $O(n^6)$ & $O(n^7)$ & $O(n^7)$\\
our algorithm          & $O(n^{3.4})$ & $O(n^4)$ & $O(n^{4.25})$ & $O(n^{4\tfrac{2}3})$ & $O(n^{5})$ & $O(n^{5.2})$
\end{tabular}
\end{center}
\caption{\label{tab:comparision}New running times for $k=5,\ldots,10$.}
\end{table}

\heading{Our Approach} 
Our algorithm applies dynamic programming on a tree decomposition. 
This is a standard method for dealing with some sparse graphs, like series-parallel graphs or outerplanar graphs.
However, in our case we work with complete graphs. 
The trick is to work on an implicit structure, called dependence graph $D$. 
Graph $D$ has $k$ vertices which correspond to the $k$ edges of $H$ that are chosen to be removed.
A subset of edges of $D$ corresponds to the pattern of edges to be added (as we will see the number of such patterns is bounded for every fixed $k$, and one can iterate over all patterns).
The dependence graph can be thought of as a sketch of the solution, which needs to be embedded in the input graph $G$.
Graph $D$ is designed so that if it has a separator $S$, such that $D-S$ falls apart into two parts $A$ and $B$, then once we find an optimal embedding of $A\cup S$ for some fixed embedding of $S$, one can forget about the embedding of $A$.
This intuition can be formalized as dynamic programming on a tree decomposition of $D$, which is basically a tree of separators in $D$.
The idea sketched above leads to an algorithm running in  time $O(n^{(1/3+\epsilon_k)k})$ for every fixed $k$, where $\lim\epsilon_k=0$.
The reason for the exponent in the running time is that $D$ is of maximum degree 4 and hence it has treewidth at most $(1/3+\epsilon_k)k$, as shown by Fomin et al.~\cite{fomin-algorithmica-2009}.

The further improvement to $O(n^{(1/4+\epsilon_k)k})$ is obtained by yet another idea.
We partition the $n$ edges of $H$ into $n^{1/4}$ buckets of size $n^{3/4}$ and we consider all possible distributions of the $k$ edges to remove into buckets.
If there are many nonempty buckets, then graph $D$ has fewer edges, because some dependencies are forced by putting the corresponding edges into different buckets.
As a result, the treewidth of $D$ decreases and the dynamic programming runs faster.
The case when there are few nonempty buckets does not give a large speed-up in the dynamic programming, but the number of such distributions is small.

\section{Preliminaries}

Throughout the paper let $w_1,w_2,\ldots,w_n$ and $e_1,\ldots,e_n$ be sequences of respectively subsequent vertices and edges visited by $H$, so that $e_i=\{w_i,w_{i+1}\}$ for $i=1,\ldots,n-1$ and $e_n=\{w_n,w_1\}$.
For $i=1,\ldots,n-1$ we call $w_i$ the {\em left endpoint} of $e_i$ and $w_{i+1}$ the {\em right endpoint} of $e_i$. Also, $w_n$ is the left endpoint of $e_n$ and $w_1$ is its right endpoint.

We work with undirected graphs in this paper. An edge between vertices $u$ and $v$ is denoted either as $\{u,v\}$ or shortly as $uv$.

For a positive integer $i$ we denote $[i]=\{1,\ldots,i\}$.

\subsection{Connection patterns and embeddings}

Formally, a $k$-move is a pair of sets $(E^-,E^+)$, both of cardinality $k$, where $E^-\subseteq \{e_1,\ldots,e_n\}$, $E^+\subseteq E(G)$, and $E(H)\setminus E^-\cup E^+$ is a Hamiltonian cycle.
This is the most intuitive definition of a $k$-move, however it has a drawback, namely it is impossible to specify $E^+$ without specifying $E^-$ first.
For this reason instead of listing the edges of $E^+$ explicitly, we will define a connection pattern, which together with $E^-$ expressed
as an {\em embedding} fully specifies a $k$-move.

A {\em $k$-embedding} (or shortly: {\em embedding}) is any function $f:[k]\rightarrow [n]$. 
A {\em connection $k$-pattern} (or shortly: {\em connection pattern}) is any perfect matching in the complete graph on the vertex set $[2k]$.
We call a connection pattern {\em valid} when one obtains a single $k$-cycle from $M$ by identifying vertex $2i$ with vertex $(2i+1)\bmod 2k$ for every $i=1,\ldots,k$.

Let us show that every pair $(E^-,E^+)$ that defines a $k$-move has a corresponding pair of an embedding and a connection pattern,
consequently giving an intuitive explanation of the above definition of embeddings and connection patterns.
Consider a move $Q=(E^-,E^+)$.
Let $E^-=\{e_{i_1},\ldots,e_{i_k}\}$, where $i_1<i_2<\cdots<i_k$.
For every $j=1,\ldots,k$, let $v_{2j-1}$ and $v_{2j}$ be the left and right endpoint of $e_{i_j}$, respectively.
An {\em embedding} of the $k$-move $Q$ is the function $f_Q:[k]\rightarrow [n]$ defined as $f_Q(j)=i_j$ for every $j=1,\ldots,k$.
Note that $f_Q$ is increasing.
A {\em connection pattern} of $Q$ is every perfect matching $M$ in the complete graph on the vertex set $[2k]$ such that $E^+=\{\{v_i,v_j\}\mid \{i,j\}\in M\}$.
Note that at least one such matching always exists, and if $E^-$ contains two incident edges then there is more than one such matching.
Note also that $M$ is valid, because otherwise after applying the $k$-move $Q$ we do not get a Hamiltonian cycle.

Conversely, consider a pair $(f, M)$, where $f$ is an increasing embedding and $M$ is a valid connection pattern.
We define $E^-_f = \{e_{f(j)}\mid j=1,\ldots,k\}$.
For every $j=1,\ldots,k$, let $v_{2j-1}$ and $v_{2j}$ be the left and right endpoint of $e_{f(j)}$, respectively.
Then we also define $E^+_{(f,M)} = \{v_iv_j\mid \{i,j\}\in M\}$.
It is easy to see that $(E^-_f,E^+_{(f,M)})$ is a $k$-move.

Because of the equivalence shown above, in what follows we abuse the notation slightly and a $k$-move $Q$ can be described both by a pair of edges to remove and add $(E^-_Q,E^+_Q)$ and by an embedding-connection pattern pair $(f_Q,M_Q)$. The {\em gain} of $Q$ is defined as $\gain(Q)=w(E^-_Q)-w(E^+_Q)$. 
Given a connection pattern $M$ and an embedding $f$, we can also define an $M$-gain of $f$, denoted by $\gain_M(f)=\gain(Q)$, where $Q$ is the $k$-move defined by $(f,M)$.
Note that \probKOPTOpt asks for a $k$-move with maximum gain.

We note that the notion of connection pattern of a $k$-move was essentially introduced by de~Berg et al.~\cite{BergBJW16} under the name of `signature', though they used a permutation instead of a matching, which we find more natural.
They also show that one can reduce the problem \probKOPTOpt so that it suffices to consider only $k$-moves where $E^-$ contains pairwise non-incident edges, but we do not find it helpful in the description of our algorithm (this assumption makes the connection pattern of a $k$-move unique).

\subsection{Tree decomposition and nice tree decomposition}

To make the paper self-contained, in this section we recall the definitions of tree and path decompositions and state their basic properties which will be used later in the paper. The content of this section comes from the textbook of Cygan et al.~\cite{fpttextbook}.
 
A {\em tree decomposition} of a graph $G$ is a pair $\mathcal{T}=(T,\{X_t\}_{t\in V(T)})$, where $T$ is a tree whose every node $t$ is assigned a vertex subset $X_t\subseteq V(G)$, called a bag,
such that the following three conditions hold:
\begin{description}
\item[(T1)] $\bigcup_{t\in V(T)} X_t =V(G)$.
\item[(T2)] For every $uv\in E(G)$, there exists a node $t$ of $T$ such that $u,v\in X_t$.
\item[(T3)] For every  $u\in V(G)$, the set $\{t\in V(T) \ |\  u\in X_t\}$ induces a connected subtree of $T$.
\end{description}
The {\em width} of tree decomposition $\mathcal{T}=(T,\{X_t\}_{t\in V(T)})$, denoted by $w(\Tt)$, equals $\max_{t\in V(T)} |X_t| - 1$. 
The {{\em treewidth}} of a graph $G$, denoted by $\tw(G)$, is the minimum possible width of a tree decomposition of $G$.
When $E$ is a set of edges and $V(E)$ the set of endpoints of all edges in $E$, by $\tw(E)$ we denote the treewidth of the graph $(V(E),E)$.

A {\em path decomposition} is a tree decomposition $\mathcal{T}=(T,\{X_t\}_{t\in V(T)})$, where $T$ is a path. 
Then $\Tt$ is more conveniently represented by a sequence of bags $(X_1,\ldots,X_{|V(T)|})$, corresponding to successive vertices of the path.
The {{\em pathwidth}} of a graph $G$, denoted by $\pw(G)$, is the minimum possible width of a path decomposition of $G$.

In what follows we frequently use the notion of {\em nice tree decomposition}, introduced by Kloks~\cite{kloks-book}. 
These tree decompositions are more structured, making it easier to describe dynamic programming over the decomposition.
A tree decomposition $\mathcal{T}=(T,\{X_t\}_{t\in V(T)})$ can be rooted by choosing a node $r\in V(T)$, called the root of $T$, which introduces a natural parent-child and ancestor-descendant relations in the tree $T$. A rooted tree decomposition $(T,\{X_t\}_{t\in V(T)})$ is {\em{nice}} if 
$X_r=\emptyset$, $X_\ell=\emptyset$ for every leaf $\ell$ of $T$, and every non-leaf node of $T$ is of one of the following three types:
\begin{itemize}
  \item{\bf{Introduce node}}: a node $t$ with exactly one child $t'$
    such that $X_t=X_{t'}\cup\{v\}$ for some vertex $v\notin X_{t'}$.
  \item{\bf{Forget node}}: a node $t$ with exactly one child $t'$ such
    that $X_t=X_{t'}\setminus \{w\}$ for some vertex $w\in X_{t'}$.
  \item{\bf{Join node}}: a node $t$ with two children $t_1,t_2$ such
    that $X_t=X_{t_1}=X_{t_2}$.
\end{itemize}

A path decomposition is nice when it is nice as tree decomposition after rooting the path in one of the endpoints. (Note that it does not contain join nodes.)

\begin{proposition}[see Lemma 7.4 in~\cite{fpttextbook}]
Given a tree (resp. path) decomposition $\mathcal{T}=(T,\{X_t\}_{t\in V(T)})$ of $G$ of width at most $k$, one can in time $O(k^2\cdot \max(|V(T)|,|V(G)|))$ compute a nice tree (resp. path) decomposition of $G$ of width at most $k$ that has at most $O(k|V(G)|)$ nodes.
\end{proposition}
 
We say that $(A,B)$ is a {\em{separation}} of a graph $G$ if $A\cup B = V(G)$ and there is no edge between $A\setminus B$ and $B\setminus A$. Then $A\cap B$ is a {\em{separator}} of this separation.
 
 \begin{lemma}[see Lemma 7.3 in~\cite{fpttextbook}]
\label{lem:tw-separate}
Let $(T,\{X_t\}_{t\in V(T)})$ be a tree decomposition of a graph $G$ and let $ab$ be an edge of $T$. The forest $T-ab$ obtained from $T$ by deleting edge $ab$ consists of two connected components $T_a$ (containing $a$) and $T_b$ (containing $b$). Let $A=\bigcup_{t\in V(T_a)} X_t$ and $B=\bigcup_{t\in V(T_b)} X_t$.
Then $(A,B)$ is a separation of $G$ with separator $X_a\cap X_b$.
 \end{lemma}

\section{The algorithm}

In this section we present our algorithms for \probKOPTOpt.
The brute-force algorithm verifies all possible $k$-moves.
In other words, it iterates over all possible valid connection patterns and increasing embeddings. 
The brilliant observation of Berg et al.~\cite{BergBJW16} is that we can iterate only over all possible connection patterns, whose number is bounded by $(2k)!$.
In other words, we fix a valid connection pattern $M$ and from now on, our goal is to find an increasing embedding $f:[k]\rightarrow [n]$ which, together with $M$, defines a $k$-move giving the largest weight improvement over all $k$-moves with connection pattern $M$.
Instead of doing this by enumerating all $\Theta(n^k)$ embeddings, Berg et al.~\cite{BergBJW16} fix carefully selected ${\lfloor 2/3 k \rfloor}$ values of $f$ in all $n^{\lfloor 2/3 k \rfloor}$ possible ways, and then show that the optimal choice of the remaining values can be found by a simple dynamic programming running in $O(nk)$ time.
Our idea is to find the optimal embedding for a given connection pattern using a different, more efficient approach. 

\subsection{Basic setup}
\label{sec:basic-setup}

Informally speaking, instead of guessing some values of $f$, we guess an {\em approximation} of $f$ defined by appropriate bucketing.
For each approximation $b$, finding an optimal embedding consistent with $b$ is done by a dynamic programming over a tree decomposition.
We would like to note that even without bucketing (i.e, by using a single trivial bucket of size $n$) our algorithm works in $n^{(1/3+\epsilon_k)k}$ time. Therefore the notion of bucketing is used to further improve the running time, but it is not essential to perform the dynamic programming on a tree decomposition.

More precisely, we partition the set $[n]$, corresponding to the edges of $H$, into buckets.
Each bucket is an interval $\{i, i+1,\ldots, j\}\subseteq [n]$, for some $1\le i \le j \le n$.
Let $n_b$ be the number of buckets and let $B_j$ denote the $j$-th bucket, for $j=1,\ldots,n_b$.
A {\em bucket assignment} is any nondecreasing function $b:[k]\rightarrow[n_b]$.

Unless explicitly modified, we use all buckets of the same size $\lceil n^\alpha\rceil$, for a constant $\alpha$ which we set later.
Then, for $j=1,\ldots,b$ the $j$-th bucket is the set $B_j = \{(j-1)\ceil{n^\alpha}+1,\ldots,j\ceil{n^\alpha}\}\cap[n]$.

Given a bucket assignment $b$ we define the set 
$$O_b=\{\{i,i+1\}\subset[k] \mid b(i)=b(i+1)\}.$$

\begin{definition}[$b$-monotone partial embedding]
Let $f:S\rightarrow [n]$ be a partial embedding for some $S\subseteq [k]$.
We say that $f$ is $b$-monotone when
\begin{itemize}
\item[(M1)] for every $i\in S$ we have $f(i)\in B_{b(i)}$, and 
\item[(M2)] for every $\{i,i+1\}\in O_b$, if $\{i,i+1\}\subseteq S$, then $f(i)<f(i+1)$.
\end{itemize}
\end{definition}

Note that a $b$-monotone embedding $f:[k]\rightarrow[n]$ is always increasing, but a $b$-monotone partial embedding does not even need to be non-decreasing (this seemingly artificial design simplifies some of our proofs).
In what follows, we present an efficient dynamic programming (DP) algorithm which, given a valid connection pattern $M$ and a bucket assignment $b$ finds a $b$-monotone embedding of maximum $M$-gain. 
To this end, we need to introduce the gain of a partial embedding.
Let $f:S\rightarrow[n]$ be a $b$-monotone partial embedding, for some $S\subseteq [k]$.
For every $j\in S$, let $v_{2j-1}$ and $v_{2j}$ be the left and right endpoint of $e_{f(j)}$, respectively.
We define
\[E^-_f = \{e_{f(i)} \mid i\in S\}\]
\[E^+_f = \{\{v_{i'},v_{j'}\} \mid i,j\in S, i'\in \{2i-1,2i\}, j'\in \{2j-1,2j\}, \{i',j'\}\in M\}.\]
Then, $\gain_M(f)=w(E^-_f)-w(E^+_f)$.

Note that $\gain_M(f)$ does not necessarily represent the {\em actual} cost gain of the choice of the edges to remove represented by $f$.
Indeed, assume that for some pair $i,j\in[k]$ there are $i'\in \{2i-1,2i\}$ and $j'\in \{2j-1,2j\}$ such that $\{i',j'\}\in M$.
Then we say that $i$ {\em interferes} with $j$, which means that we plan to add an edge between an endpoint of the $i$-th deleted edge and the $j$-th deleted edge.
Note that if $i\in S$ (the $i$-th edge is chosen) and $j\not\in S$ (the $j$-th edge is not chosen yet) this edge to be added is not known yet, and its cost is not represented in $\gain_M(f)$.
However, the value of $f(i)$ influences this cost.
Consider the following set of interfering pairs:
\[I_M = \{\{i,j\}\mid \text{$i$ interferes with $j$}\}.\]

Note that $I_M$ is obtained from $M$ by identifying vertex $2i-1$ with vertex $2i$ for every $i=1,\ldots,k$ (and the new vertex is simply called $i$).
In particular, this implies the following simple property of $I_M$.

\begin{proposition}
\label{prop:Ipi}
 Every connected component of the graph $([k], I_M)$ is a cycle or a single edge.
\end{proposition}

\subsection{Dynamic programming over tree decomposition}
\label{sec:dp}

Now we define the graph $D_{M,b}$, called {\em the dependence graph}, where  $V(D_{M,b})=[k]$ and $E(D_{M,b})=O_b\cup I_M$.
The vertices of the graph correspond to the $k$ edges to be removed from $H$ (i.e., $j$ corresponds to the $j$-th deleted edge in the sequence $e_1,\ldots,e_n$).
The edges of $D_{M,b}$ correspond to dependencies between the edges to remove (equivalently, elements of the domain of an embedding).
The edges from $O_b$ are {\em order dependencies}: edge $\{i,i+1\}$ means that the $(i+1)$-th deleted edge should appear further on $H$ than the $i$-th deleted edge.
(Note that in $O_b$ there are no edges between the last element of a bucket and the first element of the next bucket --- this is because the corresponding constraint is forced by the assignment to buckets.)
The edges from $I_M$ are {\em cost dependencies} (resulting from interference explained in Section~\ref{sec:basic-setup}).

The goal of this section is a proof of the following theorem.

\begin{theorem}
\label{thm:dp}
 Let $M$ be a valid connection $k$-pattern and let $b:[k]\rightarrow[n]$ be a bucket assignment, where every bucket is of size $\lceil n^\alpha\rceil$.
 Then, a $b$-monotone embedding of maximum $M$-gain can be found in $O(n^{\alpha(\tw(D_{M,b})+1)}k^2+2^k)$ time.
\end{theorem}

Let $\Tt=(T,\{X_t\}_{t\in V(T)})$ be a nice tree decomposition of $D_{M,b}$ with minimum width.
Such a decomposition can be found in $O^*(1.7347^k)$ time by an algorithm of Fomin and Villanger~\cite{FV2010}, though for practical purposes a simpler $O^*(2^k)$-time algorithm is advised by Bodlaender et al.~\cite{BodlaenderFKKT12}.
For every $t\in V(T)$ we denote by $V_t$ the union of all the bags in the subtree of $T$ rooted in $t$.

For every node $t\in V(T)$, and for every $b$-monotone function $f:X_t\rightarrow[n]$, we will compute the following value.
\[T_t[f] = \max_{\substack{g:V_t \rightarrow [n]\\g|_{X_t}=f\\\text{$g$ is $b$-monotone}}}\gain_M(g).\]

Then, if $r$ is the root of $T$, and $\emptyset$ denotes the unique partial embedding with empty domain, then $T_r[\emptyset]$ is the required maximum $M$-gain of a $b$-monotone embedding.
The embedding itself (and hence the corresponding $k$-move) can be also found by using standard DP techniques.
The values of $T_t[f]$ are computed in a bottom-up fashion.
Let us now present the formulas for computing these values, depending on the kind of node in the tree $T$.

\heading{Leaf node}
When $t$ is a leaf of $T$, we know that $X_t=V_t=\emptyset$, and we just put $T_t[\emptyset]=0$.

\heading{Introduce node}
Assume $X_t = X_{t'}\cup\{i\}$, for some $i\not\in X_{t'}$ where node $t'$ is the only child of $t$.
Denote $\Delta E^+_f = E^+_f\setminus E^+_{f|_{X_{t'}}}$.
Then, we claim that for every $b$-monotone function $f:X_t\rightarrow[n]$, 
\begin{equation}
\label{eq:intro}
T_t[f] = T_{t'}[f|_{X_{t'}}] + w(e_{f(i)}) - \sum_{\{u,v\}\in \Delta E^+_f}w(\{u,v\}). 
\end{equation}

We show that~\eqref{eq:intro} holds by showing the two relevant inequalities.
Let $g$ be a function for which the maximum from the definition of $T_t[f]$ is attained.
Let $g'=g|_{V_{t'}}$. Note that $g'$ is $b$-monotone because $g$ is $b$-monotone.
Hence, $\gain_M(g') \le T_{t'}[f|_{X_{t'}}]$.
It follows that $T_t[f] = \gain_M(g)=\gain_M(g') + w(e_{f(i)}) - \sum_{\{u,v\}\in \Delta E^+_f}w(\{u,v\}) \le T_{t'}[f|_{X_{t'}}] + w(e_{f(i)}) - \sum_{\{u,v\}\in \Delta E^+_f}w(\{u,v\})$. 

Now we proceed to the other inequality.
Assume $g'$ is a function for which the maximum from the definition of $T_{t'}[f|_{X_{t'}}]$ is attained.
Let $g:V_t\rightarrow [n]$ be the function such that $g|_{V_{t'}}=g'$ and $g(i)=f(i)$.
Let us show that $g$ is $b$-monotone.
The condition $(M1)$ is immediate, since $g'$ and $f$ are $b$-monotone.
For $(M2)$, consider any $\{j,j+1\}\in O_b$ such that $\{j,j+1\}\subseteq V_t$.
If $i\not\in \{j,j+1\}$ then $g(j)<g(j+1)$ by $b$-monotonicity of $g'$, so assume $i\in \{j,j+1\}$.
Then $\{j,j+1\}\subseteq X_t$, for otherwise $X_t\cap X_{t'}$ does not separate $j$ from $j+1$, a contradiction with Lemma~\ref{lem:tw-separate}.
For $\{j,j+1\}\subseteq X_t$, we have $g(j)<g(j+1)$ since $f(j)<f(j+1)$.
Hence $g$ is $b$-monotone, which implies $T_t[f]\ge \gain_M(g)$.
Then it suffices to observe that $\gain_M(g)=\gain_M(g') + w(e_{f(i)}) - \sum_{\{u,v\}\in \Delta E^+_f}w(\{u,v\})=T_{t'}[f|_{X_{t'}}] + w(e_{f(i)}) - \sum_{\{u,v\}\in \Delta E^+_f}w(\{u,v\})$.
This finishes the proof that~\eqref{eq:intro} holds.

\heading{Forget node}
Assume $X_t = X_{t'}\setminus\{i\}$, for some $i\in X_{t'}$ where node $t'$ is the only child of $t$.
Then the definition of $T_t[f]$ implies that
\begin{equation}
\label{eq:forget}
T_t[f] = \max_{\substack{f':X_{t'} \rightarrow [n]\\f'|_{X_t}=f\\\text{$f'$ is $b$-monotone}}} T_{t'}[f']. 
\end{equation}

\heading{Join node}
Assume $X_t = X_{t_1} = X_{t_2}$, for some nodes $t$, $t_1$ and $t_2$, where $t_1$ and $t_2$ are the only children of $t$.

Then, we claim that for every $b$-monotone function $f:X_t\rightarrow[n]$, 
\begin{equation}
\label{eq:join}
T_t[f] = T_{t_1}[f] + T_{t_2}[f] + \left(w(E^-_f) - w(E^+_f)\right). 
\end{equation}

Let us first show the $\le$ inequality.
Let $g$ be a function for which the maximum from the definition of $T_t[f]$ is attained.
Let $g_1=g|_{V_{t_1}}$ and $g_2=g|_{V_{t_2}}$. Note that $g_1$ and $g_2$ are $b$-monotone because $g$ is $b$-monotone.
This, together with the fact that $g_i|_{X_{t_i}}=f$ for $i=1,2$ implies $\gain_M(g_i) \le T_{t_i}[f]$ for $i=1,2$.
It follows that $T_t[f] = \gain_M(g)=\gain_M(g_1) + \gain_M(g_2) + \left(w(E^-_f) - w(E^+_f)\right) \le T_{t_1}[f] + T_{t_2}[f] + \left(w(E^-_f) - w(E^+_f)\right)$. 

Now we proceed to the $\ge$ inequality.
Assume $g_1$ (resp. $g_2$) is a function for which the maximum from the definition of $T_{t_1}[f]$ (resp. $T_{t_2}[f]$) is attained.
Let $g:V_t\rightarrow [n]$ be the function such that $g|_{V_{t_1}}=g_1$ and $g|_{V_{t_2}}=g_2$.
Note that $g|_{X_t}=f$.
Then $\gain_M(g)=\gain_M(g_1) + \gain_M(g_2) + \left(w(E^-_f) - w(E^+_f)\right) = T_{t_1}[f] + T_{t_2}[f] + \left(w(E^-_f) - w(E^+_f)\right)$.
It suffices to show that $g$ is $b$-monotone, because then $T_t[f] \ge \gain_M(g)$.
The condition $(M1)$ is immediate, since $g_1$ and $g_2$ are $b$-monotone.
For $(M2)$, consider any $\{j,j+1\}\in O_b$ such that $\{j,j+1\}\subseteq V_t$.
If $\{j,j+1\}\subseteq V_{t_1}$ or $\{j,j+1\}\subseteq V_{t_2}$ then $g(j)<g(j+1)$ by $b$-monotonicity of $g_1$ or $g_2$, respectively.
Hence, by symmetry, we can assume $j\in V_{t_1}\setminus V_{t_2}$ and $j+1\in V_{t_2}\setminus V_{t_1}$.
However, this cannot happen, because then $X_t$ does not separate $j$ from $j+1$, a contradiction with Lemma~\ref{lem:tw-separate}.

\heading{Running time}
Since $|V(T)| = O(k)$, in order to complete the proof of Theorem~\ref{thm:dp} it suffices to prove the following lemma.
\begin{lemma}
\label{lem:node-process-time}
 Let $\mathcal{T}=(T,\{X_t\}_{t\in V(T)})$ be a nice tree decomposition of $D$.
 Let $t$ be a node of $T$.
 For every $i\in X_t$ let $s_i$ be the size of the bucket assigned to $i$.
 Then, all the values of $T_t$ can be found in time $O(k\prod_{i\in X_t}s_i)$.
 In particular, if all buckets are of size $\ceil{n^\alpha}$, then $t$ can be processed in time $O(kn^{\alpha|X_t|})$.
\end{lemma}

\begin{proof}
Obviously, in every leaf node the algorithm uses only $O(1)$ time.

For an introduce node, observe that evaluation of the formula~\eqref{eq:intro} takes $O(k)$ time for every $f$, since $|\Delta E^+_f|\le 2$
(the factor $O(k)$ is needed to read off a single value from the table). 
By condition $(M1)$, each value $f(i)$ of a $b$-monotone function $f$ can be fixed in $s_i$ ways, so the number of $b$-monotone functions $f:X_t\rightarrow[n]$ is bounded by $\prod_{i\in X_t}s_i$.
Hence all the values of $T_t$ are computed in time $O(k\prod_{i\in X_t}s_i)$, which is $O(kn^{\alpha|X_t|})$ when all buckets are of size $\ceil{n^\alpha}$.

For a forget node, a direct evaluation of~\eqref{eq:forget} for all $b$-monotone functions $f:X_t\rightarrow[n]$ takes $O(k\prod_{i\in X_{t'}}s_i)$ time, where $t'$ is the only child of $t$.

Finally, for a join node a direct evaluation of~\eqref{eq:join} takes $O(k)$ time, since $|E^-_f|\le k$ and $|E^+_f|\le k$. Hence all the values of $T_t$ are computed in time $O(k\prod_{i\in X_t}s_i)$.
\end{proof}

\subsection{An algorithm running in time $O(n^{(1/3+\epsilon)k})$ for $k$ large enough}

We will make use of the following theorem due to Fomin, Gaspers, Saurabh, and Stepanov~\cite{fomin-algorithmica-2009}.

\begin{theorem}[Fomin et al.~\cite{fomin-algorithmica-2009}]
\label{thm:bound-tw-fomin}
For any $\epsilon > 0$, there exists an integer $n_\epsilon$ such that for every graph $G$ with $n>n_\epsilon$ vertices,
\[\pw(G) \le \frac{1}6 n_3 + \frac{1}3 n_4 + \frac{13}{30} n_5 + \frac {23}{45}n_6 + n_{\ge 7} + \epsilon n,\]
where $n_i$ is the number of vertices of degree $i$ in $G$ for any $i \in \{3,\ldots, 6\}$ and $n_{\ge 7}$
is the number of vertices of degree at least 7. 
\end{theorem}

We actually use the following corollary, which is rather immediate.

\begin{corollary}
\label{thm:bound-tw}
For any $\epsilon > 0$, there exists an integer $n_\epsilon$ such that for every multigraph $G$ with $n>n_\epsilon$ vertices and $m$ edges where for every vertex $v\in V(G)$ we have $2\le \deg_G(v)\le 4$, the pathwidth of $G$ is at most $(m - n)/3 + \epsilon n$. 
\end{corollary}

\begin{proof}
The corollary follows from Theorem~\ref{thm:bound-tw-fomin} by the following chain of equalities.
  \begin{equation}
  \begin{split} 
    \frac{1}6 n_3 + \frac{1}3 n_4 & = \frac{1}3\left(\frac{1}2n_3 + n_4\right) = \frac{1}3\left(\frac{1}2(2n_2+3n_3+4n_4) - (n_2 + n_3 + n_4)\right) \\
                                  & = \frac{1}3\left(\frac{1}2\sum_{v\in V(G)}\deg_G(v) - n\right) = \frac{1}3(m-n).
  \end{split}
  \end{equation}
\end{proof}

Let $P_k = \{\{i,i+1\}\mid i\in[k-1]\}$.

\begin{lemma}
\label{lem:pw(D)}
 For any $A\subseteq P_k$ we have $\pw(I_M \cup A) \le |A|/3 + \epsilon_k k$, where $\lim_{k\rightarrow\infty} \epsilon_k = 0$.
\end{lemma}

\begin{proof}
Although $([k],I_M \cup A)$ may not be of minimum degree 2, we may consider the edge multiset $I_M'$ of the graph obtained from $([k], I_M)$ by replacing every single edge component $\{u,v\}$ by a 2-cycle $uvu$.
Then $I_M'$ is a cycle cover, so every vertex in multigraph $([k],I_M' \cup A)$ has degree between 2 and 4.
Hence, by Corollary~\ref{thm:bound-tw}, for some sequence $\epsilon_k$ with $\lim_{k\rightarrow\infty} \epsilon_k = 0$ we have that 
$\pw(I_M \cup A)=\pw(I_M' \cup A)\le |I'_M|+|A|-k)/3 + \epsilon_k k \le |A|/3 + \epsilon_k k$.
\end{proof}

By Lemma~\ref{lem:pw(D)} it follows that the running time in Theorem~\ref{thm:dp} is bounded by  $O(n^{(\tfrac{\alpha}{3}+\epsilon)k})$.
If we do not use the buckets at all, i.e., $\alpha=1$ and we have one big bucket of size $n$, we get the $O(n^{(\tfrac{1}{3}+\epsilon)k})$ bound.
By iterating over all at most $(2k)!$ connection patterns we get the following result, which already improves over the state of the art for large enough $k$.

\begin{theorem}
 For every fixed integer $k$, \probKOPTOpt can be solved in time $O(n^{(1/3+\epsilon_k)k})$, where $\lim_{k\rightarrow\infty} \epsilon_k = 0$.
\end{theorem}

\subsection{An algorithm running in time $O(n^{(1/4+\epsilon)k})$ for $k$ large enough}

Let $\Mm_k$ be the set of all valid connection $k$-patterns.

\begin{lemma}
\label{lem:general-time}
 \probKOPTOpt can be solved in time $2^{O(k\log k)} n^{c(k)}$, where
 \begin{equation}
 \label{eq:c(k)}
 c(k) = \max_{M\in\Mm_k}\min_{\alpha\in[0,1]}\max_{A\subseteq P_k}\left((1-\alpha)(k-|A|)+\alpha(\tw(I_M\cup A)+1)\right).
 \end{equation}
\end{lemma}

\begin{proof}
 We perform the algorithm from Theorem~\ref{thm:dp} for each possible valid connection pattern $M$ and every bucket assignment $b$, with all the buckets of size $\lceil n^{\alpha_M}\rceil$, for some $\alpha_M\in[0,1]$.
 Let us bound the total running time.
 Let $A\subseteq P_k$ and consider a bucket assignment $b$ such that $O_b=A$.
 There are $n^{(1-\alpha_M)(k-|A|)}$ such bucket assignments, and by Theorem~\ref{thm:dp} for each of them the algorithm uses time $O(n^{\alpha_M(\tw(I_M\cup A)+1)}k^2+2^k)$.
 Hence the total running time is bounded by
 \begin{equation}
 \begin{split} 
 & \sum_{M\in\Mm_k}\sum_{A\subseteq P_k}\sum_{\substack{b:[k]\rightarrow[\ceil{n/\ceil{n^{\alpha_M}}}]\\\text{$b$ nondecreasing}\\ O_b=A}} O(n^{\alpha_M(\tw(I_M\cup A)+1)}k^2+2^k) 
   = \\
 &  O(2^k)\sum_{M\in\Mm_k}\sum_{A\subseteq P_k}n^{(1-\alpha_M)(k-|A|)}\cdot n^{\alpha_M(\tw(I_M\cup A)+1)}
 \end{split}
 \end{equation}
For every $M\in\Mm_k$, the optimal value of $\alpha_M$ can be found by a simple LP (see Section~\ref{sec:small}). The claim follows.
\end{proof}

\begin{theorem}
\label{thm:1/4}
 For every fixed integer $k$, \probKOPTOpt can be solved in time $O(n^{(1/4+\epsilon_k)k})$, where $\lim_{k\rightarrow\infty} \epsilon_k = 0$.
\end{theorem}

\begin{proof}
 Fix the same value $\alpha=3/4$ for every connection pattern $M$.
 By Lemma~\ref{lem:pw(D)} we have $(1-\alpha)(k-|A|)+\alpha(\tw(I_M\cup A)+1)\le(\tfrac{1}4+\tfrac{3}{4k}+\tfrac{3}4\epsilon_k')k$.
 The claim follows by Lemma~\ref{lem:general-time}, after putting $\epsilon_k=\tfrac{3}{4k}+\tfrac{3}4\epsilon_k'$.
\end{proof}

\subsection{Saving space}

The algorithm from Theorem~\ref{thm:1/4}, as described above, uses $O(n^{(1/4+\epsilon_k)k})$ space.
However, a closer look reveals that the space can be decreased to $O(n^{(1/8+\epsilon_k)k})$.
This is done by exploiting some properties of the specific tree decomposition of graphs of maximum degree 4, described by Fomin et al.~\cite{fomin-algorithmica-2009}, which we used in Theorem~\ref{thm:bound-tw-fomin}.

This decomposition is obtained as follows.
Let $D$ be a $k$-vertex graph of maximum degree 4.
As long as $D$ contains a vertex $v$ of degree 4, we remove $v$.
As a result we get a set of removed vertices $S$ and a subgraph $D'=D-S$ of maximum degree 3.
Then we construct a tree decomposition $\Tt'$ of $D'$, of width at most $(1/6+\epsilon_k)k$, given in the paper of Fomin and H\o{}ie~\cite{FominH06}.
The tree decomposition $\Tt$ of $D$ is then obtained by adding $S$ to every bag of $\Tt'$. 
An inductive argument (see~\cite{fomin-algorithmica-2009}) shows that the width of $\Tt$ is at most $\frac{1}3k_4+\frac{1}6k_3+\epsilon_kk$.

Assume we are given a partial $b$-monotone embedding $f_0:S\rightarrow [n]$, where $S$ is the set of removed vertices mentioned in the previous paragraph.
Consider the dynamic programming algorithm from Theorem~\ref{thm:dp}, which finds a $b$-monotone embedding of maximum $M$-gain, for a given bucket assignment $b$ and connection pattern $M$.
It is straightforward to modify this algorithm so that it computes a $b$-monotone embedding of maximum $M$-gain that extends $f_0$.
The resulting algorithm runs in time $O(n^{\alpha(\tw(D-S)+1)}k^2)$ and uses space $O(n^{\alpha(\tw(D-S)+1)})$.
Recalling that $\alpha=3/4$ and $\tw(D-S)\le (1/6+\epsilon_k)k$, we get the space bound of $O(n^{(1/8+\epsilon_k)k})$.
Repeating this for each of $n^{\alpha |S|}$ embeddings of $S$ takes time $O(n^{\alpha(|S|+\tw(D-S)+1)})$ instead of $O(n^{\alpha(\tw(D)+1)})$ from Theorem~\ref{thm:dp}.
However, as explained above, the bound on $\tw(D)$ from Theorem~\ref{thm:bound-tw-fomin} used in the proof of Theorem~\ref{thm:1/4} is also a bound on $|S|+\tw(D-S)$, so the time of the whole algorithm is still bounded by $O(n^{(1/4+\epsilon_k)k})$.

\begin{theorem}
\label{thm:1/4-space}
 For every fixed integer $k$, \probKOPTOpt can be solved in time $O(n^{(1/4+\epsilon_k)k})$ and space $O(n^{(1/8+\epsilon_k)k})$, where $\lim_{k\rightarrow\infty} \epsilon_k = 0$.
\end{theorem}

Another interesting observation is that if we build set $S$ by picking an arbitrary vertex of every edge in $O_b$, then $D':=D-S$ contains no edges of $O_b$, so it has maximum degree at most 2.
It follows that $\tw(D')\le 2$. 
Thus, in Lemma~\ref{lem:general-time} we can bound $\tw(I_M\cup A) \le |A|+2$ and for $\alpha=1/2$ we get the running time of $O(n^{k/2+3/2})$.
By using the approach of fixing all embeddings of $S$ described above, we get the space of $O(n^{\alpha\tw(D')})=O(n^{3/2})$ which is less than the $\Theta(n^2)$ space needed to store all the distances of the TSP instance.
However, the additional space can be further improved. After fixing an embedding of $S$ we find the embedding of every connected component of $D-S$ separately.
Consider such a component. 
If it is a cycle, we consider all $O(n^{\alpha})=O(n^{1/2})$ ways of fixing one of its vertices and we are left with a path, say $v_1,\ldots,v_\ell$.
The dynamic programming described in Section~\ref{sec:dp} operates on a nice path decomposition of the form $\{v_1\},\{v_1,v_2\},\{v_2\},\{v_2,v_3\},\ldots,\{v_\ell\}$.
It uses space $O(n^{2\alpha})=O(n)$ in the bags of size 2. However, by combining formulas~\eqref{eq:intro} and~\eqref{eq:forget} one can compute the DP tables for size 1 bags only, using space $O(n^{\alpha})=O(n^{1/2})$.

\begin{theorem}
\label{thm:n^1/2-space}
 For every fixed integer $k$, \probKOPTOpt can be solved in time $O(n^{k/2+3/2})$ and additional space $O(\sqrt{n})$.
\end{theorem}

We suppose that more space/time trade-offs are possible by finding small sets whose removal makes the tree decomposition somewhat small.

\subsection{Small values of $k$}
\label{sec:small}

The value of $c(k)$ in Lemma~\ref{lem:general-time} can be computed using a computer programme for small values of $k$, by enumerating all connection patterns and
using formula~(\ref{eq:c(k)}) to find optimum $\alpha$.
We used a C++ implementation (see \texttt{\small http://www.mimuw.edu.pl/\~{}kowalik/localtsp/localtsp.cpp} for the source code) including a simple $O(2^k)$ dynamic programming for computing treewidth described in the work of Bodlaender et al.~\cite{BodlaenderFKKT12}.
For every valid connection pattern $M$ our program finds the value of $\min_{\alpha\in[0,1]}\max_{A\subseteq P_k}\left((1-\alpha)(k-|A|)+\alpha(\tw(I_M\cup A)+1)\right)$ by solving a simple linear program, as follows.
\begin{align*}
\text{minimize} \quad v& & \\
\text{subject to} \quad v & \ge (1-\alpha)(k-s) + \alpha \max_{\substack{A\subseteq P_k\\|A|=s}}(\tw(I_M\cup A)+1), \quad s=0,\ldots,k-1\\
 \alpha & \in [0,1]
\end{align*}

We get running times for $k=5,\ldots,10$ described in Table~\ref{tab:small-k}.
It turns out that for $k=5,\ldots,10$ the running time does not grow when we fix the same size of the buckets $n^\alpha$ for all connection patterns, hence in Table~\ref{tab:small-k} we present also the values of $\alpha$.


\begin{table}
\begin{center}
\begin{tabular}{l|c|c|c|c|c|c}
$k$ & 5 & 6 & 7 & 8 & 9 & 10\\
\midrule
$\alpha$ & $\tfrac{2}3$ & $\tfrac{3}4$ & $\tfrac{3}4$ & $\tfrac{2}3$ & $\tfrac{4}5$ & $\tfrac{4}5$ \\
time     & $O(n^{3\tfrac{2}3})$ & $O(n^4)$ & $O(n^{4.25})$ & $O(n^{4\tfrac{2}3})$ & $O(n^{5})$ & $O(n^{5.2})$
\end{tabular}
\end{center}
\caption{\label{tab:small-k}Running times of the algorithm from Theorem~\ref{thm:1/4} for $k=5,\ldots,10$.}
\end{table}

\subsection{A refined analysis of \probFiveOPTOpt}

In this section we focus on \probFiveOPTOpt problem. This the first case where our findings may have a practical relevance, which motivates us towards a deepened analysis.
It turns out that to get the entry for $k=5$ in Table~\ref{tab:small-k} we do not need a computer, and the proof is rather short, as one can see below.

\begin{theorem}
\label{thm:3.66}
 \probFiveOPTOpt can be solved in time $O(n^{3 \frac{2}3})$.
\end{theorem}

\begin{proof}
Let $D=([5],I_M\cup A)$ be the dependence multigraph. 
Since $K_5$ is the only 5-vertex graph with treewidth larger than 3, and $D$ has at most different 9 edges, we note that $\tw(D)\le 3$.

\mycase{1: } $|A|\le 1$. 
Then either $D$ has maximum degree 2, or $D$ is a 5-cycle with a single chord.
In both cases it is easy to see that $\tw(D) \le 2$.
By Lemma~\ref{lem:general-time} this case contributes $O(n^{5(1-\alpha) + 3\alpha})=O(n^{5-2\alpha})$ to the running time.

\mycase{2: } $|A|\ge 2$.
By Lemma~\ref{lem:general-time}, this case contributes $O(n^{(5-|A|)(1-\alpha) + 4\alpha})=O(n^{3+\alpha})$ to the running time.

Putting $\alpha=2/3$ finishes the proof.
\end{proof}

One can see that the tight cases in the above proof are $|A|=0$ and $|A|=2$. A closer look at the $|A|=2$ case reveals that the source of hardness of this case is a single (up to isomorphism) graph $([5],I_M\cup A)$ of treewidth 3. It turns out that using a different bucket partition design one can save some running time in this particular case. The details are given in the proof of Theorem~\ref{thm:3.5}. However, first we need a simple technical lemma, which extends Lemma~\ref{lem:node-process-time} to general (not necessarily nice) path decompositions (it is true also for tree decompositions, but we do not need it).

\begin{lemma}
\label{lem:node-process-time-not-nice}
 Let $M$ be a valid connection $k$-pattern and let $b:[k]\rightarrow[n]$ be a bucket assignment.
 For every $i\in [k]$ let $s_i$ be the size of the bucket assigned to $i$.
 Let $(X_1,\ldots,X_r)$ be a path decomposition of $D_{M,b}$.
 Then, a $b$-monotone embedding of maximum $M$-gain can be found in time $O(rk^2\max_{t\in[r]}\prod_{i\in X_t}s_i)$.
\end{lemma}

\begin{proof}
 We create a nice path decomposition of $D$ as follows.
 For every $q=1,\ldots,r-1$ we insert between $X_{q}$ and $X_{q+1}$ a sequence of forget nodes (one for every $j\in  X_q\setminus X_{q+1}$) followed by a sequence of introduce nodes (one for every $j\in  X_{q+1}\setminus X_{q}$). 
 Thus, the resulting path decomposition has at most $rk$ nodes.
 It is clear that for each of the added forget nodes with a bag $X$, we have $\prod_{i\in X}s_i \le \prod_{i\in X_q}s_i$, and 
 for each of the added introduce nodes with a bag $X$, we have $\prod_{i\in X}s_i \le \prod_{i\in X_{q+1}}s_i$.
 The claim follows by Lemma~\ref{lem:node-process-time}.
\end{proof}

\begin{theorem}
\label{thm:3.5}
 \probFiveOPTOpt can be solved in time $O(n^{3.5})$.
\end{theorem}

\begin{proof}
We will refine the proof of Theorem~\ref{thm:3.5} by looking closer at the $|A|= 2$ case. 

\mycase{1: } $|A|=2$.
By Lemma~\ref{lem:general-time}, when $\tw(D) \le 2$, this case contributes $O(n^{3(1-\alpha) + 3\alpha})=O(n^{3})$ to the running time, so a problem arises only in case $\tw(D) = 3$.

\mycase{1.1: } The two edges of $A$ are incident.
Let $A=\{ab,bc\}$ and let $d$ and $e$ be the two vertices not incident to any edge of $A$.
We claim that $\pw(D) \le 2$.
Indeed, the sequence of bags in the desired path decomposition is $N[d]$, $(N[d]\cup N[e])\setminus\{d\}$, and $\{a,b,c\}$ when $de\in I_M$ and $N[d]$, $\{a,b,c\}$ and $N[e]$ otherwise.

\mycase{1.2: } The two edges of $A$ are not incident.
Let $A=\{ab,cd\}$ and let $e$ be the vertex not incident to any edge of $A$.
Assume $e$ is not incident with $\{a,b\}$.
Since $M$ is a perfect matching, $e$ is incident with $c$ or $d$, by symmetry assume $ec \in I_M$.
Then $\{c,d,e\}$, $N[c]\setminus\{e\}$, $\{a,b,d\}$ is a path decomposition of width 2.
Hence by symmetry we can assume $ae,ce \in I_M$. By Proposition~\ref{prop:Ipi} $a$ belongs to a cycle in $([5],I_M)$, so there are 3 subcases to consider

\mycase{1.2.1: } $ac \in I_M$. Then $I_M$ consists of the cycle $ace$ and edge $bd$. Then $\{a,c,e\}$, $\{a,b,c\}$ and $\{b,c,d\}$ is a path decomposition of width 2.

\mycase{1.2.2: } $ab \in I_M$. Then $I_M$ has one cycle $abdce$. Then $D$ is the same 5-cycle, so it has pathwidth 2.

\mycase{1.2.3: } $ad \in I_M$. Then $I_M$ has one cycle $adbce$. Note that $D$ contains a minor of $K_4$, so it has treewidth 3.
It follows that we need to modify the algorithm.
We partition the bucket containing $a$ and $b$ into $n^{\alpha/2}$ buckets of size $n^{\alpha/2}$ and we consider all possible assignments of $a$ and $b$ to these buckets.

First consider the assignments where $a$ and $b$ are in the same small bucket.
There are at most $n^{3(1-\alpha)}n^{\alpha/2}=n^{3-2.5\alpha}$ such assignments.
Consider a path decomposition of $D$ consisting of two adjacent nodes $p$ and $q$ with bags $X_p=\{a,b,c,d\}$ and $X_q=\{a,b,e\}$.
Note that each of the bags contains two vertices from a bucket of size $n^{\alpha/2}$ and at most two vertices from a bucket of size $n^{\alpha}$.
By Lemma~\ref{lem:node-process-time-not-nice} nodes $p$ and $q$ can be processed in time $O(n^{2\cdot \alpha / 2}\cdot n^{2\alpha})=O(n^{3\alpha})$.
Hence the computation for the assignments where $a$ and $b$ are in the same small bucket take $O(n^{3+\alpha/2})$ time in total.

Now consider the assignments where $a$ and $b$ are in different small buckets.
There are at most $n^{3(1-\alpha)}n^{2\alpha/2}=n^{3-2\alpha}$ such assignments.
However, the corresponding dependence graph $D'$ has one edge less than $D$, namely $E(D')=E(D)\setminus \{ab\}$.
Consider a path decomposition of $D'$ consisting of three consecutive bags $\{b,c,d\}$, $\{a,c,d\}$ and $\{a,c,e\}$.
Each of the bags contains two vertices from a bag of size $n^{\alpha}$ and one vertex from a bag of size $n^{\alpha/2}$.
By Lemma~\ref{lem:node-process-time-not-nice} each of the three nodes can be processed in time $O(n^{2\alpha+\alpha/2})=O(n^{2.5\alpha})$.
Hence the computation for the assignments where $a$ and $b$ are in different small buckets also take $O(n^{3+\alpha/2})$ time in total.

\mycase{2: } $|A| \ge 3$. 
By Lemma~\ref{lem:general-time}, this case contributes $O(n^{(5-|A|)(1-\alpha) + 4\alpha})=O(n^{2+2\alpha})$ to the running time.

To sum up, by the above and Case 1 of the proof of Theorem~\ref{thm:3.66}, the algorithm works in time $O(n^{5-2\alpha}+ n^{3+\alpha/2} + n^{2+2\alpha})$.
Putting $\alpha=3/4$ finishes the proof.
\end{proof}

The running time of Theorem~\ref{thm:3.5} can be further improved by a careful refinement of the $|A|=3$ case, as shown below.

\begin{theorem}
 \probFiveOPTOpt can be solved in time $O(n^{3.4})$.
\end{theorem}

\begin{proof}
 We will refine the proof of Theorem~\ref{thm:3.5} by looking closer at the $|A|=3$ case. 
 By Lemma~\ref{lem:general-time}, when $\tw(D) \le 2$, this case contributes $O(n^{2(1-\alpha) + 3\alpha})=O(n^{3})$ to the running time, so a problem arises only in case $\tw(D) = 3$.

\mycase{1: } $|A|=3$. 

\mycase{1.1: } The edges of $A$ form a 3-path $abcd$. Let $e$ be the vertex not incident to edges of $A$.
By Proposition~\ref{prop:Ipi} $e$ has a neighbor in $\{a,b,c,d\}$. 
By symmetry assume that $e$ has a neighbor in $\{c,d\}$. 
We partition the bucket containing $c$ and $d$ into $n^{\alpha/2}$ buckets of size $n^{\alpha/2}$ and we consider all possible assignments of $c$ and $d$ to these buckets.

First consider the assignments where $c$ and $d$ are in the same small bucket.
There are at most $n^{2(1-\alpha)}n^{\alpha/2}=n^{2-1.5\alpha}$ such assignments.
Consider the path decomposition of $D$ with two bags $\{a,b,c,d\}$ and $N[e]$.
Note that each of the bags contains at most two vertices from a bucket of size $n^{\alpha/2}$ and at most two vertices from a bucket of size $n^{\alpha}$.
By Lemma~\ref{lem:node-process-time-not-nice} each of the two nodes of path decomposition can be processed in time $O(n^{2\cdot \alpha / 2}\cdot n^{2\alpha})=O(n^{3\alpha})$.
Hence the computation for the assignments where $c$ and $d$ are in the same small bucket takes time $O(n^{2+1.5\alpha})$ in total.

Now consider the assignments where $c$ and $d$ are in different small buckets.
There are at most $n^{2(1-\alpha)}n^{2\alpha/2}=n^{2-\alpha}$ such assignments.
However, the corresponding dependence graph $D'$ has one edge less than $D$, namely $E(D')=E(D)\setminus \{cd\}$.
If $ed\in E(D')$, consider the path decomposition of $D'$ consisting of three consecutive bags $N[e], (N[d]\cup N[e])\setminus\{e\}, \{a,b,c\}$.
Otherwise, i.e., when $N_{D'}(e)=\{c\}$, consider the path decomposition $N[e], \{a,b,c\}, N[d]$.
In both cases, each of the bags is of size at most three and contains at least one vertex from a bucket of size $n^{\alpha/2}$.
By Lemma~\ref{lem:node-process-time-not-nice} each of the three nodes of path decomposition can be processed in time $O(n^{2\alpha+\alpha/2})=O(n^{2.5\alpha})$.
Hence the computation for the assignments where $c$ and $d$ are in different small buckets takes $O(n^{2+1.5\alpha})$ time in total.

\mycase{1.2: } Graph $([5],A)$ has two connected components: a single edge $ab$ and a 2-path $cde$.
Note that $N(a)\cap \{c,d,e\}=N(b)\cap \{c,d,e\}=\{c,e\}$ contradicts Proposition~\ref{prop:Ipi}.
It follows that one of the following four cases holds: $N(a)\cap \{c,d,e\} \subseteq \{c,d\}$, $N(a)\cap \{c,d,e\} \subseteq \{d,e\}$, $N(b)\cap \{c,d,e\} \subseteq \{c,d\}$, $N(b)\cap \{c,d,e\} \subseteq \{d,e\}$.
Hence, by symmetry, we can assume the first of them, i.e., $N(a)\cap \{c,d,e\} \subseteq \{c,d\}$.

We partition the bucket containing $c$, $d$ and $e$ into $n^{\alpha/3}$ buckets of size $n^{\frac{2}{3}\alpha}$.

First we generate all assignments where $c$ and $d$ are in the same small bucket.
There are at most $n^{2(1-\alpha)}n^{\alpha/3}=n^{2-\frac{5}{3}\alpha}$ such assignments.
In DP we use the following path decomposition: $\{b,c,d,e\}, N[a]$.
Note that $N[a]\setminus\{c,d\} = \{a,b\}$, so each of the bags contains at most two vertices from a bucket of size $n^{\frac{2}{3}\alpha}$ and two vertices from a bucket of size $n^{\alpha}$.
By Lemma~\ref{lem:node-process-time-not-nice} each of the three nodes of path decomposition can be processed in time $O(n^{2\cdot \frac{2}{3}\alpha}\cdot n^{2\alpha})=O(n^{\frac{10}{3}\alpha})$.
Hence the computation for the assignments where $c$ and $d$ are in the same small bucket takes time $O(n^{2+\frac{5}{3}\alpha})$ in total.

Here we branch into two subcases.

\mycase{1.2.1: } $N_D(c)=\{a,b,d\}$.
Then we generate all remaining bucket assignments, i.e., where $c$ and $d$ are in different small buckets.
There are at most $n^{2(1-\alpha)}n^{2\alpha/3}=n^{2-\frac{4}{3}\alpha}$ such assignments.
In the new dependence graph $D'$ we have $E(D')=E(D)\setminus \{cd\}$.
By Proposition~\ref{prop:Ipi} and our assumptions $N(a)\cap \{c,d,e\} \subseteq \{c,d\}$ and $N(c)=\{a,b,d\}$, we get that either $A$ has two components, namely $A=\{ab,bc,ca,de\}$ or $A$ is a single 5-cycle $A=\{ad,de,eb,bc,ca\}$. 
In both cases we use the path decomposition $\{b,d,e\}, \{a,b,d\}, \{a,b,c\}$.
Each of the bags is of size at most three and contains at least one vertex from a bucket of size $n^{\frac{2}{3}\alpha}$.
By Lemma~\ref{lem:node-process-time-not-nice} each of the three nodes of path decomposition can be processed in time $O(n^{2\alpha+\frac{2}{3}\alpha})=O(n^{\frac{8}{3}\alpha})$.
Hence the computation for the assignments where $c$ and $d$ are in different small buckets takes $O(n^{2+\frac{4}{3}\alpha})$ time in total.

\mycase{1.2.2: } $N_D(c)\ne\{a,b,d\}$.
We continue by generating all assignments where $d$ and $e$ are in the same small bucket.
There are at most $n^{2(1-\alpha)}n^{\alpha/3}=n^{2-\frac{5}{3}\alpha}$ such assignments.
In the new dependence graph $D'$ we have $E(D')=E(D)\setminus \{de\}$.
In DP we use the following path decomposition: $N[c]\cup\{d,e\}, \{a,b,d,e\}$.
Note that each of the bags contains two vertices from a bucket of size $n^{\frac{2}{3}\alpha}$ and at most two vertices from a bucket of size at most $n^{\alpha}$.
By Lemma~\ref{lem:node-process-time-not-nice} each of the three nodes of path decomposition can be processed in time $O(n^{2\cdot \frac{2}{3}\alpha}\cdot n^{2\alpha})=O(n^{\frac{10}{3}\alpha})$.
Hence the computation for the assignments where $d$ and $e$ are in the same small bucket takes time $O(n^{2+\frac{5}{3}\alpha})$ in total.

Finally, we generate all assignments where $c$, $d$ and $e$ are in three different small buckets.
There are at most $n^{2(1-\alpha)}n^{3\alpha/3}=n^{2-\alpha}$ such assignments.
In the new dependence graph $D'$ we have $E(D')=E(D)\setminus \{cd,de\}=I_M \cup \{ab\}$.
By Proposition~\ref{prop:Ipi}, $I_M$ is a 5-cycle or a 3-cycle and a single edge (not incident to the cycle).
Hence $D'$ is an outerplanar graph, and hence it has a tree decomposition of width 2.
In this decomposition every bag has size a most 3 and if it has size 3, then it contains at least one vertex from $\{c,d,e\}$.
Hence every bag $B$ contains at least $|B|-2$ vertices from a bucket of size $n^{\frac{2}{3}\alpha}$.
By Lemma~\ref{lem:node-process-time-not-nice} each of the three nodes of path decomposition can be processed in time $O(n^{2\alpha+\frac{2}{3}\alpha})=O(n^{\frac{8}{3}\alpha})$.
Hence the computation for the assignments where $c$, $d$ and $e$ are in three different small bucket takes time $O(n^{2+\frac{5}{3}\alpha})$ in total.

\mycase{2: } $|A| = 4$. 
By Lemma~\ref{lem:general-time}, this case contributes $O(n^{(1-\alpha) + 4\alpha})=O(n^{1+3\alpha})$ to the running time.

To sum up, by Case 1 of the proof of Theorem~\ref{thm:3.66}, Case 1 of the proof of Theorem~\ref{thm:3.5} and Cases 1 and 2 above, the algorithm works in time $O(n^{5-2\alpha}+ n^{3+\alpha/2}+ n^{2+\frac{5}{3}\alpha}+ n^{1+3\alpha})$.
Putting $\alpha=4/5$ finishes the proof.
 \end{proof}


\section{Lower bound for $k=4$}

In this section we show a hardness result for \probOPTOpt{4}. 
More precisely, we work with the decision version, called \probOPTDec{4}, where the input is the same as in \probOPTOpt{4} and the  goal is to determine if there is a $4$-move which improves the weight of the given Hamiltonian cycle.
To this end, we reduce the \probNegativeTriangle problem, where the input is an undirected, complete graph $G$, and a weight function $w:E(G)\rightarrow \mathbb{Z}$.
The goal is to determine whether $G$ contains a triangle whose total edge-weight is negative.

\begin{lemma}
\label{lem:reduction}
	Every instance $I=(G, w)$ of \probNegativeTriangle can be reduced in $O(|V(G)|^2)$ time into
	an instance $I'=(G', w', C)$ of \probOPTDec{4} such that $G$ contains a triangle of negative weight iff $I'$ admits an improving $4$-move.
	Moreover, $|V(G')| = O(|V(G)|)$, and the maximum absolute weight in $w'$ is larger by a constant factor than the maximum absolute weight in $w$.
\end{lemma}

\begin{proof}
	\newcommand{\Vup}{V_{\mathsf{up}}}
	\newcommand{\Vdown}{V_\mathsf{down}}
	Let $V(G) = \{v_1, \ldots, v_n\}$.
	Then let $\Vup = \{a_1, b_1, \ldots, a_n, b_n\}$,
	$\Vdown = \{a_1', b_1', \ldots, a_n', b_n'\}$
	and $V(G') = \Vup\ \dot{\cup}\ \Vdown$.
	Let $W$ be the maximum absolute value of a weight in $w$.
	Then let $M_1 = 5W + 1$ and $M_2 = 21M_1 + 1$
	and let
	\[
		w'(u, v) = \begin{cases}
			0 & \text{if } (u, v) \text{ is of the form }
				(a_i, b_i')\\
			w(v_i, v_j) & \text {if } (u, v)
				\text{ is of the form }
				(a_i, b_j) \text{ for } i < j
				\text{ or } (a_i', b_j) \text { for } j < i\\
			M_1 & \text{if } (u, v) \text{ is of the form }
				(a_i, b_i)\\
			-3M_1 & \text{if } (u, v) \text{ is of the form }
				(a_i', b_i')\\
			-M_2 & \text{if } (u, v) \text{ is of the form }
				(b_i, a_{i + 1})
				\text{ or } (b_i', a_{i + 1}')
				\text{ or } (a_1, a_1')
				\text{ or } (b_n, b_n')\\
			M_2 & \text{in other case}.
		\end{cases}
	\]
	Note that the cases are not overlapping. (Note also that although some weights are negative, we can get an equivalent instance with nonnegative weights by adding $M_2$ to all the weights.)
	Let $C = {a_1, b_1, \ldots, a_n, b_n, b_n', a_n', \ldots, b_1', a_1'}$.
	
	\begin{figure}
		\begin{center}
			\usetikzlibrary{decorations.pathmorphing}
\usetikzlibrary{positioning}

\tikzstyle{vertex}=[
  circle,
  draw,
  fill=white,
  inner sep=0pt,
  minimum width=16pt,
]

\tikzstyle{label}=[
  midway,
]

\tikzstyle{fixedpath}=[
  very thick,
  decorate,
  decoration=snake,
]

\tikzstyle{fixededge}=[
	very thick,
]

\tikzstyle{removededge}=[
  ultra thick,
  dotted,
  red,
]

\tikzstyle{possibleedge}=[
	semithick,
	dashed,
]

\tikzstyle{usededge}=[
	ultra thick,
	dashed,
	blue,
]

\begin{tikzpicture}[very thick, scale=1.1, yscale=0.8]
\draw [fixedpath](-0.5,0.5) .. controls (-2.5,-1) and (-7,0) .. (-7.5,3);

\node [vertex] (a_i) at (-6.5,6) {$a_i$};
\node [vertex] (b_i) at (-5.5,6.5) {$b_i$};
\node [vertex] (a_j) at (-3.5,7) {$a_j$};
\node [vertex] (b_j) at (-2.5,7) {$b_j$};
\node [vertex] (a_k) at (-0.5,6.5) {$a_k$};
\node [vertex] (b_k) at (0.5,6) {$b_k$};

\node [vertex] (b_k') at (0.5,1) {$b_k'$};
\node [vertex] (a_k') at (-0.5,0.5) {$a_k'$};

\node [vertex] (b_n) at (1.5,4) {$b_n$};
\node [vertex] (b_n') at (1.5,3) {$b_n'$};
\node [vertex] (a_1) at (-7.5,4) {$a_1$};
\node [vertex] (a_1') at (-7.5,3) {$a_1'$};

\draw[fixedpath] (b_i) -- (a_j);
\draw [fixedpath] (b_j) -- (a_k);
\draw [fixedpath](b_k) -- (b_n);
\draw [fixedpath](b_n') -- (b_k');
\draw [fixedpath](a_1) -- (a_i);

\draw [fixededge] (a_1') edge (a_1);
\draw [fixededge] (b_n') edge (b_n);

\draw [removededge] (a_i) -- (b_i) node[label, above left=2pt and -6pt] {$M_1$};
\draw [removededge] (a_j) -- (b_j) node[label, above=2pt] {$M_1$};
\draw [removededge] (a_k) -- (b_k) node[label, above right=2pt and -4pt] {$M_1$};
\draw [removededge] (b_k') -- (a_k')  node[label, below right=2pt and -9pt] {$-3M_1$};

\draw [usededge] (a_k) -- (b_k') node[label, left] {0};
\draw [usededge] (a_i) to [bend right] node[label, below right] {$w(v_i, v_j)$} (b_j);
\draw [usededge] (a_j) to [bend right] node[label, below] {$w(v_j, v_k)$} (b_k);
\draw [usededge] (a_k') -- (b_i) node[label, left] {$w(v_i, v_k)$};
\end{tikzpicture}
			\caption{A simplified view of the instance
				$(G', w', C)$ together
				with an example of a $4$-move.
				The added edges are marked as blue (dashed) and the
				removed edges are marked as red (dotted).
			}
		\end{center}
	\end{figure}
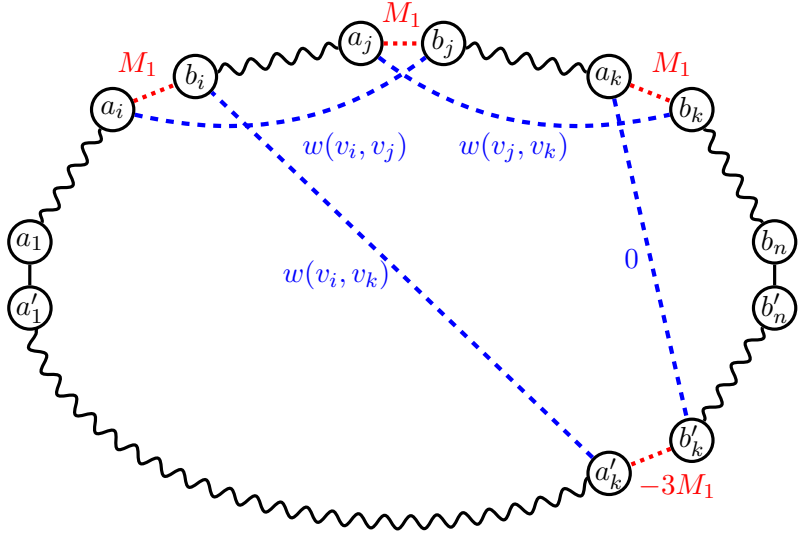

	If there is a negative triangle $v_i, v_j, v_k$ for
	some $i < j < k$ in $G$
	then we can improve $C$ by removing edges
	$(a_i, b_i), (a_j, b_j), (a_k, b_k)$ and $(a_k', b_k')$
	and inserting edges
	$(a_i, b_j), (a_j, b_k), (a_k, b_k')$ and $(a_k', b_i)$.
	We obtain a cycle
	\[
	  a_1, b_1, \ldots a_i,
	  b_j, a_{j + 1} \ldots, a_k,
	  b_k', a_{k + 1}', \ldots, b_n',
	  b_n, a_n, \ldots, b_k,
	  a_j, b_{j - 1}, \ldots b_i,
	  a_k', b_{k - 1}', \ldots, a_1'.
	\]
	The total weight of the removed edges is
	$M_1 + M_1 + M_1 + (-3M_1) = 0$
	and the total weight of the inserted edges is
	$w(v_i, v_j) + w(v_j, v_k) + 0 + w(v_k, v_i) < 0$
	hence indeed the cycle is improved.
	
	Let us assume that $C$ can be improved by removing $4$ edges
	and inserting $4$ edges.
	Note that all the edges of weight $-M_2$ belong to $C$
	and all the edges of weight $M_2$ do not belong to $C$.
	All the other edges have absolute values of their weights
	bounded by $3M_1$.
	Therefore even a single edge of the weight $-M_2$ cannot
	be removed
	and even a single edge of the weight $M_2$ cannot be
	inserted because a loss of $M_2$ cannot be compensated
	by any other $7$ edges (inserted or removed), as they can result in a gain of at most $7\cdot 3 M_1<M_2$.
	Hence in the following we treat edges of weights $\pm M_2$
	as fixed, i.e., they cannot be inserted or removed from the 
	cycle.
	Note that the edges of $C$ that can
	be removed are only the edges of the form $(a_i, b_i)$
	(of weights $M_1$)
	and $(a_i', b_i')$ (of weights $-3M_1$).
	
	All the edges of weight $-3M_1$ already belong
	to $C$ and all the remaining edges of the graph
	that can be inserted or removed from the cycle
	are the edges
	of the weight $M_1$ belonging to
	$C$ and the edges of absolute values of
	their weights bounded by $W.$
	Therefore we
	cannot remove more than one edge of the weight $-3M_1$ from $C$
	because a loss of $6M_1$
	cannot be compensated by any $2$ removed and $4$
	inserted edges (we could potentially gain only
	$2M_1 + 4W < 3M_1$). Hence we can remove at most one
	edge of the weight $-3M_1$
	from $C$.
	For the same reason if we do remove one edge of the weight
	$-3M_1$  (i.e., of the form $(a_i', b_i')$)
	from $C$ we need to remove also three edges
	of the weights $M_1$ (i.e., of the form $(a_j, b_j)$)
	in order to compensate the loss
	of $3M_1$
	(otherwise we could compensate up to $2M_1 + 5W < 3M_1)$.
	
	Note that the only edges that can be added
	(i.e., the edges with the weights less than $M_2$ that do not
	belong to $C$)
	are
	the edges of the form $(a_i, b_j)$ for $i < j$, 
	$(a_i', b_j)$ for $j < i$ and $(a_i, b_i')$.
	Therefore
	if the removed edges from $G[\Vup]$ are
	$(a_{i_1}, b_{i_1}), \ldots, (a_{i_\ell}, b_{i_\ell})$
	for some $i_1 < \ldots < i_\ell$
	(and no other edges belonging to $G[\Vup]$)
	then in order to close the cycle
	we need to insert some edge incident to $b_{i_1}$
	but since for any $i_0 < i_1$
	there is no removed edge $(a_{i_0}, b_{i_0})$
	it cannot be an edge of the form
	$(a_{i_0}, b_{i_1})$.
	Hence it has to be an edge of the form $(a_j', b_{i_1})$
	for some $j > i_1$.
	But then also the edge $(a_j', b_j')$ has to be removed.
	Therefore if we remove at least one edge of the form
	$(a_i, b_i)$ then we need to remove also an edge of the
	form $(a_j', b_j')$ (and as we know this implies also that
	at least three edges of the form $(a_i, b_i)$ have to be
	removed).
	So if any edge is removed, then exactly three edges
	of the form $(a_i, b_i)$ and exactly one edge of the form
	$(a_j', b_j')$ have to be removed.
	Note that this implies also that the total weight of the removed
	edges has to be equal to zero.
	
	Clearly the move has to remove at least one edge in order
	to improve the weight of the cycle.
	Let us assume that the removed edges are
	$(a_i, b_i)$, $(a_j, b_j)$ and $(a_k, b_k)$
	for some $i < j < k$
	and $(a_\ell', b_\ell')$ for some $\ell$.
	For the reason mentioned in the previous paragraph
	in order to obtain a Hamiltonian cycle one of the inserted edges has to be the edge $(a_\ell', b_i)$.
	Also the vertex $b_j$ has to be connected with something
	but the vertex $a_\ell'$ is already taken
	and hence it has to be connected with the vertex $a_i$.
	Similarly the vertex $b_k$ has to be connected with
	$a_j$ because $a_\ell'$ and $a_i$ are already taken.
	Thus $a_k$ has to be connected with $b_\ell'$
	and this means that $k = \ell$.
	The total weight change of the move is negative
	and therefore the total weight of the added
	edges has to be negative
	(since the total weight of the removed edges is
	equal to zero).
	Thus we have
	$w(v_i, v_j) + w(v_j, v_k) + w(v_k, v_i)
	= w'(a_i, b_j) + w'(a_j, b_k) + w'(a_k', b_i) + w'(a_k, b_k')
	< 0$.
	So $v_i, v_j, v_k$ is a negative triangle in $(G, w)$.
\end{proof}

\begin{theorem}
  If there is $\epsilon>0$ such that \probOPTDec{4} admits an algorithm in time $O(n^{3-\epsilon} \cdot \polylog(M))$, then there is $\delta>0$ such that both \probNegativeTriangle and \probAPSP admit an algorithm in time $O(n^{3-\delta} \cdot \polylog(M))$, where in all cases we refer to $n$-vertex input graphs with integer weights from $\{-M, \ldots, M\}$.
\end{theorem}

\begin{proof}
The first part of the claim follows from Lemma~\ref{lem:reduction}, while the second part follows from the reduction of \probAPSP to \probNegativeTriangle by Vassilevska-Williams and Williams (Theorem 1.1 in~\cite{WilliamsW10}).
\end{proof}

\bibliographystyle{abbrv}
\bibliography{localtsp}

\end{document}